\documentclass[conference,letterpaper]{IEEEtran}

\usepackage[utf8]{inputenc} 
\usepackage[T1]{fontenc}
\usepackage{url}
\usepackage{ifthen}
\usepackage{cite}
\usepackage[cmex10]{amsmath} 

\usepackage{amsfonts, amsmath, amssymb, amsthm}
\usepackage{array, makecell}
\usepackage{tikz, color, soul}
\usepackage{pgfplots}
\usetikzlibrary{plotmarks}
\usepackage{mathtools}

\newtheorem{theorem}{Theorem}
\newtheorem{proposition}{Proposition}
\newtheorem{lemma}[theorem]{Lemma}

\newtheorem{example}{Example}

\theoremstyle{definition}
\newtheorem{definition}{Definition}

\theoremstyle{remark}
\newtheorem{remark}{Remark}
\newtheorem{construction}{Construction}

\newcommand{\cC}{{\mathcal C}}

\newcommand{\N}{{\mathbb N}}

\newcommand{\F}{\mathbb{F}}

\renewcommand{\epsilon}{\varepsilon}

\newcommand{\pr}[1]{\operatorname{Pr}\left({#1}\right)}
\newcommand{\cpr}[2]{\operatorname{Pr}\left({#1}\mid{#2}\right)}
\newcommand{\expect}[1]{\operatorname{E}{#1}}

\begin{document}
	\title{Constructions of Batch Codes via Finite Geometry}
	\author{%
		\IEEEauthorblockN{\textbf{Nikita~Polyanskii}\IEEEauthorrefmark{1}, and \textbf{Ilya~Vorobyev}\IEEEauthorrefmark{1}\IEEEauthorrefmark{2} }
		\IEEEauthorblockA{\IEEEauthorrefmark{1}Center for Computational and Data-Intensive Science and Engineering, \\
			Skolkovo Institute of Science and Technology\\
			Moscow, Russia 127051}
		\IEEEauthorblockA{\IEEEauthorrefmark{2}Advanced Combinatorics and Complex Networks Lab, \\
			Moscow Institute of Physics and Technology\\  Dolgoprudny, Russia 141701}
		\IEEEauthorblockA{\textbf{Emails}: nikita.polyansky@gmail.com, vorobyev.i.v@yandex.ru}
	}
	
	\maketitle
	
	\begin{abstract}
		A primitive $k$-batch code encodes a string $x$ of length $n$ into string $y$ of length $N$, such that each multiset of $k$ symbols from $x$ has $k$ mutually disjoint recovering sets from $y$.  We develop new explicit and random coding constructions of linear primitive batch codes based on finite geometry. In some parameter regimes, our proposed codes have lower redundancy than previously known batch codes. 
	\end{abstract}
 

	\section{Introduction}
	Batch codes were originally proposed by Ishai et al.~\cite{ishai2004batch} for
	load balancing in distributed systems, and amortizing the computational cost of private information retrieval and related cryptographic protocols. Ishai et al. gave a definition of \textit{batch codes} in a general form, namely $n$ information symbols $x_1,\ldots,x_n$ are encoded to an $m$-tuple of strings $y_1,\ldots, y_m$ (referred to as \textit{buckets}) of total length $N$, such that for each $k$-tuple (\textit{batch}) of distinct indices $i_1,\ldots, i_k\in[n]$, the entries $x_{i_1},\ldots, x_{i_k}$ can be decoded by reading at most $t$ symbols from each bucket. The parameter $k$ is usually called \textit{availability} and it plays an important role in supporting high throughput of the distributed storage system. If a batch could contain any \textit{multiset} of indices (not only distinct indices), then we use the term a \textit{multiset batch code}. In a special case when $t=1$ and each bucket contains one symbol, a multiset  batch code is called \textit{primitive}. This class of batch codes is the most studied one in the literature since there are several statements~\cite{ishai2004batch} which allow to trade between different choices of $n$, $N$, $m$, $t$ and $k$. In other words, better constructions of primitive batch codes would imply better constructions of multiset batch codes. 
	\subsection{Notation}
	We denote the field of size $2$ by $\F_2$.  The symbol $[n]$ stands for the set of integers
	$\{1, 2,\ldots, n\}$. Let us give a formal definition of codes studied in this
	paper.
	\begin{definition}\label{definition batch code}
		Let $C$ be a linear code of length $N$ and dimension $n$ over the field $\F_2$, which encodes a string $x_1,\ldots,x_n$ to $y_1,\ldots, y_N$.
		The code $C$ will be called a \textit{primitive linear $k$-batch code} (simply, \textit{$k$-batch code}), and will be denoted	by $[N, n, k]^B$, if for every multiset of symbols $\{x_{i_1},\ldots, x_{i_k} \},$ ${i_j}\in	[n]$, there exist $k$ mutually disjoint sets $R_{i_1},\ldots , R_{i_k}\subset[N]$ (referred to as \textit{recovering sets}) such that for all $j \in [k]$, $x_{i_j}$ is a sum of the symbols $y_p$ with indices $p$ from $R_{i_j}$.
	\end{definition}
	Given $n$ and $k$, we denote the minimal integer $N$ such that an $[N,n,k]^B$ code exists by $N_B(n,k)$. In this paper we focus on the minimal redundancy of batch codes, which we abbreviate by $r_B(n,k):=N_B(n,k)-n$. 
	
	Recall that a \textit{systematic linear code} is a linear code in which the input data is embedded in the encoded output, i.e., $y_i=x_i$ for $i\in[n]$. In what follows we are going to construct systematic linear batch codes.
	The following special case of recovering sets will be particularly useful.
	\begin{definition}\label{simple recovering set}
		For a systematic linear code, we say that the recovering set $R$ for information symbol $x_i$ is \textit{simple} if $R$ contains exactly one index greater than $n$. In other words, if $j$ is such an index, then 
		$$
		y_j = x_i + \sum_{t\in R\setminus \{j\}}x_t.
		$$
	\end{definition}
	Note that many constructions, suggested earlier and in this paper, possess a more stronger property than one described in Definition~\ref{definition batch code} -- the existence of mutually disjoint simple recovering sets. 
	
	We use the notation $n^{\epsilon^-}$ in a statement to demonstrate that the statement remains true for all $n^{\epsilon-c}$, where $c$ is any fixed positive number. In the rest of the paper we will mainly concentrate on the case when $k=n^\epsilon$, $n\to\infty$.
	\subsection{Related Work}
	The
	authors of~\cite{ishai2004batch} provided constructions of various families	of batch codes. Those constructions were based on unbalanced expanders, on recursive application of
	trivial batch codes, on smooth and Reed-Muller
	codes, and others. Many other constructions proposed later in~\cite{asi2018nearly, rawat2016batch, vardy2016constructions} improve the redundancy of batch codes. In particular, a systematic linear code, defined by the generator matrix $G=[I_n|E]$, is shown~\cite{rawat2016batch} to be a $k$-batch code, where $k$ is the minimal number of ones in rows of $E$ and the bipartite graph, whose biadjacency matrix is $E$, has no cycle of length at most $6$. Constructions based on array codes and multiplicity codes were investigated in~\cite{asi2018nearly}.
	
	There is another class of related codes which is called \textit{combinatorial batch codes}. For these codes the same property as for the batch codes is required, but symbols cannot be encoded. Such codes were investigated in~\cite{bhattacharya2012combinatorial, brualdi2010combinatorial,silberstein2016optimal, stinson2009combinatorial,bujtas2011combinatorial}.	A special case of batch codes, called \textit{switch codes}, was studied in~\cite{wang2013codes,buzaglo2018consecutive,chee2015combinatorial,wang2015optimal}. It was suggested in~\cite{wang2013codes}
	to use such codes to increase the parallelism of data routing in the network switches. 
	\textit{Private information retrieval (PIR) codes}
	can be seen as an instance of batch codes, namely we require a weaker property that every information symbol has $k$ mutually independent recovering sets. PIR codes were suggested in~\cite{fazeli2015pir} to decrease storage overhead in PIR schemes preserving both privacy and communication complexity. Some constructions and bounds for PIR codes can be found in~\cite{blackburn2017pir,rao2016lower,asi2018nearly,fazeli2015pir,zhang2016private}. \textit{One-step majority-logic decodable codes}~\cite{lin2001error} require a stronger property than PIR codes, namely every encoded symbol should have $k$ mutually independent recovering sets.  Also we refer the reader to \textit{locally repairable codes with availability}~\cite{wang2015achieving,pamies2013locally,rawat2016locality}, which have an additional (with respect to PIR codes) constraint on the size of recovering sets.
	
	Recall some known results on the minimal redundancy of batch codes:
	\begin{enumerate}
		\item $r_B(n,k)\ge k-1$;
		\item $r_B(n,k)=\Omega(\sqrt{n})$ for $k\ge 3$, \cite{wootters2016, rao2016lower};
		\item $r_B(n,k)=O(k^2 \sqrt{n} \log n)$ for $k\le \sqrt{n}/\log n$, \cite{asi2018nearly};
		\item $r_B(n,n^\epsilon)\le n^{7/8}$ for $7/32<\epsilon\le1/4$, \cite{rawat2016batch};
		\item $r_B(n,n^\epsilon)\le n^{4\epsilon}$ for $1/5<\epsilon\le7/32$, \cite{rawat2016batch};
		\item $r_B(n,n^{\epsilon^-}) \le n^{5/6 + \epsilon/3}$ for $0 < \epsilon \le 1/2$, \cite{asi2018nearly};
		\item $r_B(n,n^{1-\epsilon})\le n^{1-\delta}$ for $0\le\epsilon\le 1$, where ${\delta = \delta(\epsilon) > 0}$,~\cite{asi2018nearly}.
	\end{enumerate}
	In particular, it follows that the best known lower bound on the redundancy of batch codes is as follows
	\begin{equation}\label{lower bound on r}
	r_B(n,k)\ge \Omega(\max(\sqrt{n}, k)).
	\end{equation}
	\subsection{Our contribution}
	In this paper we develop new explicit and random coding constructions of linear primitive batch codes based on finite geometry. In Table~\ref{summary}  our contribution (upper bounds on $r_B(N,k)$) is summarized.	
	\begin{table}[h]
		\caption{Binary Batch Codes Summary}
		\label{summary}
		\begin{center}
			\begin{tabular}{|c|c|c|}
				\hline 
				Construction & Availability $k$ & Redundancy $r_B(n,k)$ \\ 
				\hline 
				Theorem~\ref{random construction} (random)  & $k=o(n^{1/3}/\log n)$ &  $O(k^{3/2}\sqrt{n}\log n)$\\ 
				\hline 
				\makecell{Theorem~\ref{construction in dimensional space} (explicit) \\ for any $\ell \in \N$} & $k<\frac{1}{\ell^2}n^{1/{(2\ell+1)}}$ & $O\left(kn^{\frac{\ell+1}{2\ell+1}}\right)$\\
				\hline
			\end{tabular}
		\end{center}
	\end{table} 
	
	Let us denote $r_B(n,k=n^\epsilon)=:O(n^{\delta})$. The lower bound given by~\eqref{lower bound on r} along with old and new upper bounds on $\delta = \delta(\epsilon)$ are plotted in Figure~\ref{asymptotic results}. The existence result of our work shows that the known upper bound on $\delta(\epsilon)$ can be improved for $\epsilon\in(0,2/7)\setminus\{1/5,1/4\}$.  Furthermore, we emphasize that the endpoints of novel explicit constructions by Theorem~\ref{construction in dimensional space} lye on the segment given by the random construction in Theorem~\ref{random construction}. 
	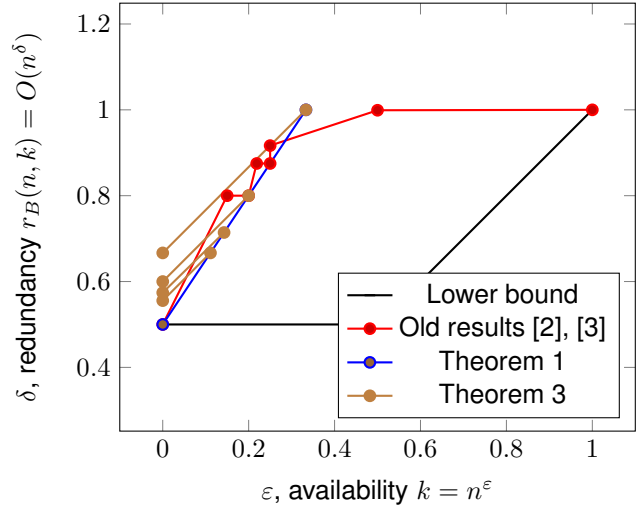
\begin{figure}
		\begin{center}
			\begin{tikzpicture}[y=7cm, x=7cm,font=\sffamily]
			\begin{axis}[axis equal, 
			xtick={0,0.2,...,1.2},
			ytick={0.2,0.4,...,1.2},
			legend pos=south east,
			xlabel={$\epsilon$, availability $k=n^\epsilon$}, ylabel={$\delta$, redundancy $r_B(n,k)=O(n^\delta)$}]			
			\addplot+[mark=-, color = black, thick] plot coordinates
			{(0,	0.5)	(0.5, 0.5) (1, 1)};
			\addlegendentry{Lower bound};
			\addplot+[mark=*, color = red, thick] plot coordinates
			{(0,	0.5) (0.15,	0.8) (0.2,	0.8) (0.21875,	0.875)	(0.25,	0.875)		(0.25,	0.91666666)	(0.5,	0.999)	(1,	0.99999999) };
			\addlegendentry{Old results \cite{asi2018nearly,rawat2016batch}}
			\addplot+[mark=*, color = blue, thick] plot coordinates
			{ (0,	0.5)	(0.3333333, 1) };
			\addlegendentry{Theorem \ref{random construction}}
			\addplot[mark=*, color = brown, thick] plot coordinates
			{(0,	0.6666666)	(0.3333333, 1)};
			\addlegendentry{Theorem \ref{construction in dimensional space}}
			\addplot[mark=*, color = brown, thick] plot coordinates
			{(0,	0.6)	(0.2, 0.8)};
			\addplot[mark=*, color = brown, thick] plot coordinates
			{(0,	0.57412)	(0.1428, 0.7142)};
			\addplot[mark=*, color = brown, thick] plot coordinates
			{(0,	0.55555)	(0.11111, 0.6666666)};
			\end{axis}				
			\end{tikzpicture}  
			\caption{Asymptotic results for binary primitive batch codes}
			\label{asymptotic results}
		\end{center}
	\end{figure}
	\subsection{Outline}
	The remainder of the paper is organized as follows. In Section~\ref{random construction section} we prove the existence of batch codes using the probabilistic method. The achieved upper bound on the redundancy improves previously known results when $k=n^\epsilon$ and $\epsilon\in(0,2/7)\setminus\{1/5, 1/4\}$. We note that for $k=n^{1/4}$ and $k=n^{1/5}$, the redundancy of our construction is worse by the multiplicative factor $\log n$ than one in~\cite{rawat2016batch}.   In Section~\ref{explicit construction section} we describe our main results and give new  explicit constructions of batch codes. In a more detail, we associate information bits with elements of vector space $\F_q^{2\ell+1}$, $\ell\in\N$, and define parity-check bits as sums of information bits lying in some affine $\ell$-dimensional subspaces. 
	Finally, Section~\ref{conclusion} concludes the paper. 
	\section{Random Construction of Batch Codes}\label{random construction section}	
	To prove the following statement, we consider a systematic linear code defined by the generator matrix $G=[I_n|E]$, where $E$ is taken as an incidence matrix of randomly chosen family of subsets of lines in the affine plane. 
	\begin{theorem}\label{random construction}
		For $k=o(n^{1/3}/\log n)$, the redundancy of $k$-batch codes is
		$$
		r_B(n,k)= O(k^{3/2}\sqrt{n}\log n).
		$$
	\end{theorem}
	\begin{proof}
		For simplicity of notation and without loss of generality, we assume that $n=q^2$, $q$ is a prime power integer and $k<q/12$. Consider a finite affine plane $(P,L)$ of order $q$, where $P$, $|P|=n$, is a set of points, and $L$, $|L|=n+q$, is a set of lines. Each line is known to contain $q$ points, and each point is in $q+1$ lines, any two lines in the affine plane cross each other in at most $1$ point.
		
		Let us randomly choose a family $F:=\{S_1,\ldots, S_M\}$ of subsets of lines in the affine space. First, we take each line in the affine space independently with probability $p_1$, which will be specified later. Second, we define a subset of any included line by leaving each point on the line independently with probability $p_2$, which will be specified later. It can be seen that for a proper choice $p_1$, the cardinality of $F$, $|F|=M$ (total number of subsets), is ``close'' to its average $p_1(n+q)$ with high probability, and for a proper choice of $p_2$, the cardinality of any subset $S_i$ is ``close'' to its average $p_2 q$. We define event $W_1$ when the total number of lines $M>3p_1n$, and $W_2$ if there exists some $S_i$ of size $>2p_2 q$. Moreover, we define $W_{2,j}$, $j\in[n]$, if there exists $S_i$ of size $>2p_2 q$ such that the line corresponding to subset $S_i$ does not contain the $j$th point.
		
		Now we consider some bijection between $n$ information symbols and $n$ points. Therefore, the information symbols are associated with the points in the plane.  Given a subset $S_i$, we can define a parity-check symbol $y_i$ as a sum of information symbols corresponding to points in $S_i$.		
		Let us consider a systematic linear code $C$ of length $n+M$ and dimension $n$ defined as a map $\phi:\F_2^n \to \F_2^{n+M}$:
		$$
		\phi(x_1,\ldots, x_n):=(x_1,\ldots, x_n,y_1,\ldots, y_M).
		$$
		Given a multiset of information symbols of size $k$, we can  uniquely represent it in the form 
		$$((x_{i_1},k_1),\ldots, (x_{i_\ell},k_\ell)),
		$$
		where 
		$$
		1\le i_1<\ldots < i_\ell \le n \text{ and }\sum_{i=1}^{\ell}k_i = k.
		$$
		We define a greedy algorithm for constructing a collection of recovering sets for any given multiset of information bits of size at most $k$. Assume that the algorithm can construct simple recovering sets for the multiset 
		$$
		((x_{i_1},k_1),\ldots,(x_{i_{j-1}},k_{j-1})), \quad j-1<\ell,
		$$
		representing the first $j-1$ groups of the multiset
		$$((x_{i_1},k_1),\ldots, (x_{i_\ell},k_\ell)).
		$$
		Then find first $k_{j}$  parity-check symbols depending on symbol $x_{i_j}$, such that the corresponding $k_{j}$ simple recovering sets are  disjoint with already chosen recovering sets, and $k_{j}$ lines corresponding to the parity-check symbols does not go through any point in the set 
		$$
		\{x_{i_1},\ldots, x_{i_{j-1}},x_{i_{j+1}},\ldots, x_{i_{\ell}}\}.
		$$ Let us add these $k_{j}$ recovering sets to the collection of recovering sets. We note that added $k_{j}$ simple recovering sets are mutually disjoint by our construction.
		
		To show that the code $C$ is likely to be a $k$-batch code, we are going to estimate the probability of event $B$ that the greedy algorithm fails for some multiset of information symbols. To get an estimate of this event, we introduce auxiliary terminology.		 We say that the information symbol $x_i$  is \textit{$s$-bad}, $0\le s<k$, if there exists some multiset 
		\begin{align*}
		((x_{i_1},k_1),\ldots,(x_{i_\ell},k_\ell)) \text{ with } i=i_{j},\,i_1<\ldots<i_\ell&,\\ \sum\limits_{f\in[\ell]\setminus\{j\}}k_f=s,\, s+k_{j}=k&,
		\end{align*}
		so that the algorithm finds recovering sets for the first $(j-1)$ groups of the multiset and fails to find $k_{j}$ recovering sets for $x_i=x_{i_{j}}$. 
		Let $B_{i,s}$ be an event that information symbol $x_i$ is $s$-bad. If no event among $B_{i,s}$ occurs, then the event $B$ doesn't happen.
		
		We note that $k$-batch code with redundancy at most $3p_1n$ exists if $\Pr(B\cup W_1)<1$. 
		Now we estimate this event as follows
		\begin{multline}\label{big bad event}
		\Pr(B\cup W_1)\le \Pr(W_1)+\Pr\left(\bigcup\limits_{\substack{i\in[n]\\ s\in\{0,\ldots,k-1\}}}B_{i,s}\right) \\\leq \pr{W_1} + \pr{W_2} + kn\max\limits_{\substack{i\in[n]\\ s\in\{0,\ldots,k-1\}}}\pr{B_{i,s}\cap\overline{W}_{2}}. 
		\end{multline}
		
		It is easy to estimate $\pr{W_1}$ and $\pr{W_2}$ applying the Chernoff bound in the form
		$$
		\Pr(X\ge(1+\delta)\mu)\le e^{-\frac{\delta^2 \mu}{3}},
		$$
		where $0<\delta < 1$, and $X$ is a sum of independent random variables taking values in $\{0,1\}$ with $\expect X = \mu$. We have
		\begin{multline}\label{estimate W1}
		\pr{W_1}=\pr{M> 3p_1n}\\
		\le\pr{M> 2p_1(n+q)}\le e^{-\frac{p_1n}{3}} 
		\end{multline}
		and
		\begin{multline}\label{estimate W2}
		\pr{W_2}=\pr{\text{``there exists $S_i$ of size $>2p_2 q$''}}\\
		\le 2ne^{-\frac{p_2 q}{3}}.
		\end{multline}
		Now we estimate the third probability in~\eqref{big bad event} as follows
		\begin{multline}\label{union bound bad}
		\pr{B_{i, s}\cap\overline{W}_{2}}\le \pr{B_{i, s}\cap\overline{W}_{2,i}}\\\le n^{k-1}\pr{A\cap C \cap \overline{W}_{2,i}}	\le
		n^{k-1}\cpr{A}{\overline{W}_{2,i} \cap C},
		\end{multline}
		where $C$ stands for the event that the algorithm  finds recovering sets 
		$$
		R_{i_1,1}, \ldots, R_{i_1,k_{i_1}},R_{i_2,1}\ldots, R_{i_{j-1}, k_{i_{j-1}}}
		$$
		for the first $j-1$ groups of 
		$$
		((x_{i_1},k_1),\ldots,(x_{i_{\ell}},k_{j_\ell})),
		$$
		and $A$ denotes the event that the algorithm fails to find $k-s$ recovering sets for $x_i$, which are disjoint with all recovering sets the algorithm found. Let $I_1;=\{i_1,\ldots, i_\ell\}$, and 
		$I_2$ be a set of information symbols included to recovering sets   
		$$
		R_{i_1,1}, \ldots, R_{i_1,k_{i_1}},R_{i_2,1}\ldots, R_{i_{j-1}, k_{i_{j-1}}}. 
		$$
		The cardinality of $I_2$ given the event $\overline{W}_{2,i}$ (consequently, given the event $\overline{W}_{2,i}\cap C$) is upper bounded as follows
		\begin{equation}\label{second info set}
		|I_2|=\sum_{u=1}^{j-1}\sum_{v=1}^{k_u}(|R_{i_{u},v}|- 1)\le 2 q p_2k,
		\end{equation}
		since $\overline{W}_{2,i}$ stands for the event that all the subsets corresponding to the lines disjoint with $x_i$ are of size at most $2p_2 q$. 
		The total number of lines containing $x_i$ is equal to $q+1$. One can easily see that there are at most $k$ of them which have a nonzero intersection with $I_1$. Since all the lines containing fixed point $x_i$ share only $x_i$, we claim that there are at most $q/2$ lines which intersect $I_2$ by at least $4p_2k$ points. Indeed, otherwise we can lower bound the cardinality of $I_2$ by $\ge 4p_2k(q /2 + 1)$ which contradicts with \eqref{second info set}. 
		We shall try to recover symbol $x_i$ with the help of other $t$, $t\ge q/2 - k\ge q/3$, lines. Enumerate them from $1$ to $t$. Let $\xi_1, \ldots, \xi_{t}$ be indicator random variables, which equals 1 iff 
		\begin{enumerate}
			\item the corresponding line was randomly taken (with probability $p_1$),
			\item the symbol $x_i$ was left (with probability $p_2$) and included to the parity-check sum,
			\item none of the symbols from $I_2$ were added in the corresponding parity-check.
		\end{enumerate}		
		
		Define the random variable 
		$$
		\eta:=\sum\limits_{i=1}^{q/3}\xi_i.
		$$
		Since $\xi_i$ is an independent Bernoulli random variable with probability $p'_i\geq p_1p_2(1-p_2)^{4p_2j}$, we claim that Binomial random variable $\chi$ with parameters $q/3$ and $p_1p_2(1-p_2)^{4p_2j}$ is stochastically dominated by $\eta$. Now we proceed with upper bounding~\eqref{union bound bad}  as follows
		\begin{multline*}
		\cpr{A}{\overline{W}_{2,i}\cap C}\le \pr{\chi < k-s} \\\le  \binom{q/3}{k}\left(1-p_1p_2(1-p_2)^{4p_2k}\right)^{q/3-k} \\\le q^k\left(1-p_1p_2(1-p_2)^{4p_2k}\right)^{q/4}.		
		\end{multline*}
		Combining the last inequality together with~\eqref{big bad event}-\eqref{union bound bad} yields
		\begin{multline}\label{final estimate}
		\Pr(B\cup W_1)\\\le 2ne^{-\frac{p_2 q}{3}}+e^{-\frac{p_1n}{3}}+kn^k q^k(1-p_1p_2(1-p_2)^{4p_2k})^{q/4}.
		\end{multline}
		Given $\epsilon>0$, there exists sufficiently large $q_0$ such that for $q>q_0$ the first two terms are at most $\epsilon$. Now we proceed with the last term
		$$
		kn^k q^k(1-p_1p_2(1-p_2)^{4p_2k})^{q/4}\le kn^{1.5k} e^{-p_1p_2(1-p_2)^{4p_2k}q/4}.
		$$
		Taking $p_2:=1/\sqrt{8k}$, we have $4p_2k\ge 1$ and 
		$$
		(1-p_2)^{4p_2k}\ge 1-  4p_2^2k=1/2.
		$$
		From this it follows that for
		$$
		p_1 :=36\frac{k^{3/2}\log n}{\sqrt{n}}
		$$
		and sufficiently large $n$, $n=q^2$, the last term in~\eqref{final estimate} is at most $\epsilon$. Therefore, we obtain that there exists a $k$-batch code with redundancy $M<108k^{3/2}\sqrt{n}\log n$ with probability at least $1-3\epsilon$.  This completes the proof.
	\end{proof}
	\section{Explicit Construction of Batch Codes}\label{explicit construction section}
	In this section to construct batch codes we associate information bits with elements of vector space $\F_q^{2\ell+1}$, $\ell\in\N$, and define parity-check bits as sums of information bits lying in some affine $\ell$-dimensional subspaces. In particular, the following finite geometry framework turns out to be useful.
	
	\begin{definition}\label{nice def}
		Suppose $\{V_1,\ldots, V_m\}$ is a collection of $\ell$-dimensional subspaces in $\F_q^{2\ell+1}$. This collection is said to be \textit{$L$-nice} if  the two properties hold:
		\begin{enumerate}
			\item any two distinct subspaces from this collection have the trivial intersection in the origin only, i.e. $|V_i\cap V_j|=1$ for $i\neq j$; 
			\item for all $i\in[m]$ and for all $v\in\F_q^{2\ell+1}$, $v\not\in V_i$, the affine subspace $v+V_i$ intersects at most $L$ subspaces from this collection.
		\end{enumerate}
	\end{definition}
	Such a framework appears to be new in the literature up to our best knowledge. In the following statement we show how to use a nice collection of subspaces to construct batch codes.
	\begin{lemma}\label{connection}
		Suppose $\{V_1,\ldots, V_m\}$ is an $L$-nice collection of $\ell$-dimensional subspaces in $\F_q^{2\ell+1}$.   Then there exists a $[q^{2\ell+1}+ m q^{\ell+1},q^{2\ell+1}, \lfloor m / L \rfloor ]^B$ code.
	\end{lemma}
	We postpone the proof of Lemma~\ref{connection} to Appendix.
	Now we give a construction of nice subspaces, which represents a collection of Reed-Solomon codes of length $2\ell+1$ and dimension $\ell$.
	\begin{construction}\label{construction}
		Let $V$ stand for a $(2\ell+1)$-dimensional $\F_q$-vector space, and  $B$ 
		is an $\F_q$-basis for $V$.  Now  let us define a collection $\cC$ of subspaces of size $m:=\lfloor q/\ell\rfloor$.  Let the $i$th, $0\le i< m$, subspace $V_i\in \cC$ be the linear span of $\ell$ vectors $\{{v}^i_1,\ldots, {v}^i_\ell\}$, where vector ${v}^i_j$, $j\in\{0,\ldots,\ell-1\}$, is written in basis $B$ as follows
		$$
		{v}^i_j:= (1, \alpha^{\ell i+j},\alpha^{2(\ell i+j)},\ldots,  \alpha^{{2\ell(\ell i + j)}}).
		$$
	\end{construction}
	We prove that $\cC$ is $\ell$-nice in Proposition~\ref{lower bound on card}. Let $m(L,\ell,q)$ be the maximal number $m$ such that  there exists an $L$-nice collection of $\ell$-dimensional subspace in $\F_q^{2\ell + 1}$ of cardinality $m$.  The next two propositions establish a quite tight estimate on the maximal cardinality of a nice collection of subspaces.
	\begin{proposition}\label{lower bound on card}
		Construction~\ref{construction} is $\ell$-nice. This implies, in particular, for any  $\ell,\,L\in\N$, $L\ge \ell$, and prime power integer $q$, the lower bounds on $m(L,\ell,q)$ holds
		$$
		m(L,\ell,q) \ge \lfloor q / \ell \rfloor.
		$$
	\end{proposition}
	
	\begin{proposition}\cite{wootters2018} \label{upper bound on card}
		For any  $\ell,\,L\in\N$ and prime power integer $q$, the upper bounds on $m(L,\ell,q)$ holds
		$$
		m(L,\ell,q) \le (L+1)q.
		$$
	\end{proposition}
	We postpone the proof of Proposition~\ref{lower bound on card} to Appendix. The proof of Proposition~\ref{upper bound on card}, suggested by Mary Wootters, is included to Appendix for completeness of the paper.
	
	Finally Lemma~\ref{connection} and Proposition~\ref{lower bound on card} imply the following upper bound on the redundancy of batch codes.
	\begin{theorem}\label{construction in dimensional space}
		For any $\ell\in \N$, prime power integer $q$ and integer $k$, $0<k\le \lfloor q/\ell^2\rfloor$, the redundancy of $k$-batch codes is upper bounded by
		$$
		r_B(n,k) \le \ell k q^{\ell+ 1},
		$$
		where $n=q^{2\ell+1}$.
	\end{theorem}
	\begin{remark}
		Proposition~\ref{upper bound on card} verifies that the proposed framework based on finite geometry could not be significantly improved  in terms of the range of parameter $k$ in Theorem~\ref{construction in dimensional space}, that is $k$ could not be larger than $\lfloor(L+1)q/L\rfloor$. 
	\end{remark}
	\begin{proof}[Proof of Theorem~\ref{construction in dimensional space}]
		From Proposition~\ref{lower bound on card} it follows that there exists  an $\ell$-nice collection of $\ell$-dimensional subspaces in $\F_q^{2\ell+1}$, which has cardinality $\lfloor q / \ell \rfloor$. Take any subset of this collection of size $m=\ell k$, where $k\le \lfloor q / \ell^2 \rfloor$.   Lemma~\ref{connection} states that there exists a $[q^{2\ell + 1} + \ell kq^{\ell+1},q^{2\ell + 1}, k]^B$ code. This completes the proof.
	\end{proof}
	Let us demonstrate how Theorem~\ref{construction in dimensional space} actually works.
	\begin{example} Let  $q=3$, $\ell = 1$ and $k=3$. Then $n=3^3=27$. Denote by $\F_3 = \{0,1,2\}$. Let us index $n$ information symbols by vectors of $\F_3^3$, i.e., $x_{000},x_{001},\ldots,x_{222}$. First we define three direction vectors $(1, 0, 0)$, $(1,1,1)$ and $(1,2,1)$, which are linearly independent. We shall construct a systematic linear code. One can determine $kn^{2/3}=27$ parity-check bits as sums of information bits which indexes lye on lines with given direction vectors. These lines represent distinct $1$-dimensional affine subspaces of $\F_3^3$. For instance, there are $9$ lines with direction vector $(1,2,1)$. Let us take one which goes through point $(0,1,2)$. Then the corresponding parity-check bit is $y_{i'}=x_{012}+x_{100} + x_{221}$ and the recovering set for $x_{012}$ based on this parity-check bit is $\{i', 100, 221\}$. It is easy to show that there are $2$ other simple recovering sets for $x_{012}$, which are of the form $\{i'',112, 212\}$ and $\{i''', 120, 201\}$. Moreover, each information bit has exactly $3$ simple recovering sets. For every bit, each of its recovering sets has a nonempty intersection with at most one recovering set of any other bit. This property immediately implies~\cite{rawat2016batch} that our code is a $3$-batch code. For $\ell>1$, in the proof of Lemma~\ref{connection} we will show a generalization of this property.
	\end{example}	
	\section{Conclusion}\label{conclusion}
	In this paper new random coding bound and new explicit constructions of primitive linear batch codes based on finite geometry were developed. In some parameter regimes, our codes improves the redundancy than previously known batch codes. We note that the random coding bound coincides with the constructive bound in a countable number of points and gives better result in others. 	The natural open question arose in this work is to construct codes which would achieve random coding bound in all others points too. Another interesting question is how to improve the lower bound given by inequality~\eqref{lower bound on r}.  
	
	\section*{Acknowledgment}
	We thank Eitan Yaakobi for the fruitful discussion on batch codes and Mary Wootters for the proof of Proposition~\ref{upper bound on card}.
	N. Polyanskii was supported in part the Russian Foundation for Basic Research (RFBR) through grant nos. 18-07-01427~A, 18-31-00310~MOL\_A. I. Vorobyev was supported in part by RFBR through grant nos. 18-07-01427~A, 18-31-00361~MOL\_A.
	\bibliographystyle{IEEEtran}
	\bibliography{batch}
	\appendix \label{proofs}
	\begin{proof}[Proof of Lemma~\ref{connection}]
		Let $\cC := \{V_1,\ldots, V_m\}$ be an $L$-nice collection of $\ell$-dimensional subspaces of $\F_q^{2\ell+1}$. Now we construct a systematic linear code of length $N:=q^{2\ell+1}+mq^{\ell+1}$ and dimension $n:=q^{2\ell+1}$. First, we associate $n$ information symbols with $n$ vectors in $\F_q^{2\ell+1}$. For every affine subspace of form $v+V_i$, $v\in\F_q^{2\ell+1}$, $V_i\in\cC$, we define a parity-check symbol as a sum of information symbols lying in this affine subspace. One can easily see that each information symbol is involved in $m$ different parity-checks. From this and the fact that for any $v\in\F_q^{2\ell+1}$, any two distinct affine subspaces $v+V_i$ and $v+V_j$ have only trivial intersection in $v$, it follows that each information symbol has $m$ mutually disjoint simple recovering sets. The total number of parity-check bits is 
		$$
		\frac{|\F_q^{2\ell+1}||\cC|}{|V_i|}=mq^{\ell+1}.
		$$
		Since collection $\cC$ is $L$-nice, we conclude that for every bit, each of its simple recovering sets has a nonempty intersection with at most $L$ simple recovering set of any other bit. Therefore, for any multiset request of information symbols of size at most $\lfloor m/L \rfloor$, we are able to construct simple recovering sets. This completes the proof.
	\end{proof}
	\begin{remark}
		We emphasize that the property we used in the proof of Lemma~\ref{connection} allows to construct recovering sets in an arbitrary order, that is a multiset request of information symbols could be given bit-by-bit and the corresponding recovering sets could be output bit-by-bit also.
	\end{remark}
	\begin{proof}[Proof of Proposition~\ref{lower bound on card}]
		Let us prove that collection $\cC$ is $\ell$-nice. 
		
		First, we need to show property $1)$ in Definition~\ref{nice def}.  Indeed if two distinct subspaces $V_i$ and $V_j$, $i\ne j$, have a nontrivial intersection, then the matrix   
		$$
		\begin{bmatrix}
		1& \alpha^{\ell i}&\alpha^{2(\ell i)}&\ldots&  \alpha^{{2\ell(\ell i)}} \\
		1& \alpha^{\ell i+1}&\alpha^{2(\ell i+1)}&\ldots&  \alpha^{{2\ell(\ell i+1)}} \\
		\hdotsfor{5} \\
		1& \alpha^{\ell i+\ell-1}&\alpha^{2(\ell i+\ell-1)}&\ldots&  \alpha^{{2\ell(\ell i+\ell-1)}} \\
		1& \alpha^{\ell j}&\alpha^{2(\ell j)}&\ldots&  \alpha^{{2\ell(\ell j)}} \\
		1& \alpha^{\ell j+1}&\alpha^{2(\ell j+1)}&\ldots&  \alpha^{{2\ell(\ell j+1)}} \\
		\hdotsfor{5} \\
		1& \alpha^{\ell j+\ell-1}&\alpha^{2(\ell j+\ell-1)}&\ldots&  \alpha^{{2\ell(\ell j+\ell-1)}} \\
		\end{bmatrix}
		$$
		is of rank $<2\ell$. However, this is a Vandermonde matrix with distinct elements in the second column, and thus it has maximum rank $2\ell$. Additionally, we see that all vector subspaces in $\cC$ are $\ell$-dimensional.
		
		Second, let us check property $2)$  in Definition~\ref{nice def}.
		Seeking a contradiction, suppose $v+V_{i_0}$ overlaps with $\ell+1$ distinct vector subspaces $V_{i_1}, V_{i_2},\ldots,  V_{i_{\ell+1}}$. It follows that vector $v$ belongs to $\ell+1$ subspaces $V_{i_0}\oplus V_{i_1}, V_{i_0}\oplus V_{i_2},\ldots, V_{i_0}\oplus V_{i_{\ell+1}}$, where $\oplus$ is the direct sum of subspaces. We know that vector $\mathbf{v}$ does not belong to $V_{i_0}$. Thus, $2\ell$-dimensional subspaces $V_{i_0}\oplus V_{i_u}$, $u\in[\ell+1]$, are intersected by an $(\ell+1)$-dimensional subspace $V'$. Let $g_u\in \F_q^{2\ell+1}$, $u\in[\ell+1]$, be the vector whose coordinates are coefficients (from the constant coefficient  to the leading coefficient) of the polynomial
		\begin{multline*}
		f_u(x) = (x-\alpha^{\ell i_0})(x-\alpha^{\ell i_0+1})\ldots (x-\alpha^{\ell i_0+\ell-1})	\\\times(x-\alpha^{\ell i_u})(x-\alpha^{\ell i_u+1})\ldots (x-\alpha^{\ell i_u+\ell-1}).
		\end{multline*}
		Since the inner products $\langle g_u,v^{i_u}_j\rangle = f_u(\alpha^{\ell i_u+j})= 0  $ and $\langle g_u,v^{i_0}_j\rangle = f_u(\alpha^{\ell i_0 + j})= 0$, we get that $g_u$ is orthogonal to $V_{i_0}\oplus V_{i_u}$.
		The condition $\dim V' = \ell+1$ is equivalent to a linear dependency of the system of vectors $\{g_u, \,u\in[\ell+1]\}$. Therefore, there exist non-trivial coefficients $c_u\in \F_q$, $u\in[\ell+1]$, so that the linear combination 
		$$
		\sum_{u=1}^{\ell+1}c_u f_u(x)=0
		$$
		vanishes.
		Since all $f_u(x)$, $u\in[\ell+1]$, share the same polynomial $(x-\alpha^{\ell i_0})\ldots (x-\alpha^{\ell i_0+\ell-1})$ and the size of the field is sufficiently large, we conclude
		$$
		\sum_{u=1}^{\ell+1}c_u	(x-\alpha^{\ell i_u})(x-\alpha^{\ell i_u+1})\ldots (x-\alpha^{\ell i_u+\ell-1}) = 0.
		$$
		This equality yields that either the matrix
		$$
		\begin{bmatrix}
		1& \alpha^{\ell i_1}&\alpha^{2(\ell i_1)}&\ldots&  \alpha^{{\ell(\ell i_1)}} \\
		1& \alpha^{\ell i_2}&\alpha^{2(\ell i_2)}&\ldots&  \alpha^{{\ell(\ell i_2)}}\\
		\hdotsfor{5} \\
		1& \alpha^{\ell i_{\ell+1}}&\alpha^{2(\ell i_{\ell+1})}&\ldots&  \alpha^{{\ell(\ell i_{\ell+1})}} \\
		\end{bmatrix}
		$$
		is singular or there is at least one zero coefficient of polynomial $g_{\alpha}(x)=(x-1)(x-\alpha)\ldots(x-\alpha^{\ell-1})$. However, the matrix is a Vandermonde matrix with distinct elements in the second column, whereas $g_\alpha(x)$ is a generator polynomial of Reed-Solomon code and all its coefficients are nonzero. This contradiction completes the proof.
	\end{proof} 
	\begin{proof}[Proof of Proposition~\ref{upper bound on card}]
		Let $\cC:=\{V_1,\ldots, V_m\}$ be an $L$-nice collection of $\ell$-dimensional subspaces in $\F_q^{2\ell+1}$. For each $i$, let $G_i\in \F^{(2\ell+1)\times\ell}_q$ be a matrix whose
		columns span $V_i$, and let $H_i\in \F^{(\ell+1)\times (2\ell+1)}_q$ be a matrix whose rows span $V_i^{\perp}$. Since $\cC$ is $L$-nice, we
		have the following property: for all $i\in [m]$, for all $v\in \F^{2\ell+1}_q$, $v\not\in V_i$, the number of $j$ so that $V_j$ intersects
		$v+V_i$ is at most $L$.
		
		Fix some (arbitrary) $i$. For $j\in[m]\setminus \{i\}$, notice that $H_iG_j\in \F^{(\ell+1)\times\ell}_q$ has full column rank because $V_i$ and
		$V_j$ do not intersect non-trivially, so let $g_j$ be a nonzero vector in the one-dimensional subspace perpendicular
		to the column span of $H_iG_j$, that is $g^T_j H_iG_j = 0$. Let $G\in \F^{(m-1)\times(\ell+1)}_q$ be the matrix with the $g_j$'s as its
		rows.
		
		Next, we shall prove that for any vector $w\in \F^{\ell+1}_q$, $Gw$ has at most $L$ zeros. To see this, suppose that $\langle g_j ,w\rangle = 0$,
		that is, the $j$th element of $Gw$ is zero. By definition, this means that $w$ is in the column span of $H_iG_j$ , say
		that $w = H_iG_jy$ for some $y\in \F^\ell_q$. Let $v\in F^{2\ell+1}_q$ be such that $H_iv = w$, so we have
		$$
		H_iG_jy = H_iv
		$$
		which means that
		$$
		G_jy = v + G_ix
		$$
		for some $x$ (that is, $G_jy$ and $v$ differ by something in the kernel of $H_i$ which is the image of $G_i$). But this
		means precisely that $V_i + v$ and $V_j$ intersect. Since there are at most $L$ values of $j$ so that $V_i + v$ and $V_j$
		intersect, we conclude that there are at most $L$ values of $j$ so that the $j$th element of $Gw$ is zero, which
		proves the claim.
		
		Thus, whenever we have a collection of $m$ subspaces with the desired property, we have a matrix $G\in \F^{(m-1)\times(\ell+1)}_q$ so that every $L + 1$ rows of $G$ have rank $\ell + 1$; otherwise there would be some $w$ in the kernel
		of those rows that result in a vector $Gw$ with too many zeros. We claim that we must have $m\le (L + 1)q$ for
		such a matrix $G$ to exist. Indeed, let $W$ be a random subspace of dimension $\ell$ in $F^{\ell+1}_q$ , and observe that
		$$
		\expect|W \setminus \{g_1, \ldots, g_{m-1}\}| =
		\sum_{j=1}^{m-1} \Pr \{g_j \in W\} =
		\frac{m-1}{q}
		$$
		so if $m\ge(L + 1)q + 1$, there exists some subspace $W$ so that
		$|W \setminus  \{g_1,\ldots,g_{m-1}\}|\ge   L + 1$. This contradiction completes the proof. 
	\end{proof}
\end{document}